\documentclass[11pt]{article}

\usepackage{theorem}
\usepackage{subfigure}
\usepackage{epsfig}
\usepackage{amssymb}
\usepackage{fullpage}
\usepackage{amsmath}
\usepackage{ifthen}
\usepackage{verbatim}
\usepackage{hyperref}
\usepackage{xspace}
\usepackage{bm}
\usepackage{setspace}
\usepackage{tablefootnote}

\newtheorem{theorem}	 			{Theorem}[section]

{\theorembodyfont{\rmfamily} \newtheorem{definition}
[theorem]	{Definition}}
{\theorembodyfont{\rmfamily} \newtheorem{remark}		[theorem]
{Remark}}
{\theorembodyfont{\rmfamily} \newtheorem{example}		[theorem]
{Example}}
{\theorembodyfont{\rmfamily} }
{\theorembodyfont{\rmfamily} }
{\theorembodyfont{\rmfamily} }
{\theorembodyfont{\rmfamily} }
{\theorembodyfont{\rmfamily} }
{\theorembodyfont{\rmfamily} }
{\theorembodyfont{\rmfamily} }
\theoremstyle{break}
{\theorembodyfont{\rmfamily} }

\newenvironment{proof}{\noindent {\em {Proof:}}}{$\blacksquare$\vskip
\belowdisplayskip}

\newenvironment{prevproof}[2]{\noindent {\em {Proof of
{#1}~\ref{#2}:}}}{$\blacksquare$\vskip \belowdisplayskip}

\newcommand{\argmax}{\operatornamewithlimits{argmax}}

\newcommand{\prob}[2][]{\text{\bf Pr}\ifthenelse{\not\equal{}{#1}}{_{#1}}{}\!\left[#2\right]}
\newcommand{\expect}[2][]{\text{\bf E}\ifthenelse{\not\equal{}{#1}}{_{#1}}{}\!\left[#2\right]}

\def\sm{\setminus}
\def\sse{\subseteq}
\def\eps{\epsilon}

\newcommand{\tfm}{(\allocs,\prices,\burn)}

\newcommand{\bid}{b}
\newcommand{\bids}{{\mathbf \bid}}
\newcommand{\bidsmt}{{\mathbf \bid}_{-t}}
\newcommand{\bidt}[1][t]{{\bid_{#1}}}

\newcommand{\val}{v}
\newcommand{\vals}{{\mathbf \val}}

\newcommand{\valt}[1][t]{{\val_{#1}}}

\newcommand{\alloc}{x}
\newcommand{\allocs}{{\mathbf \alloc}}

\newcommand{\alloct}[1][t]{\alloc_{#1}}

\newcommand{\price}{p}
\newcommand{\prices}{{\mathbf \price}}
\newcommand{\pricet}[1][t]{\price_{#1}}

\newcommand{\burn}{q}

\newcommand{\burnt}[1][t]{\burn_{#1}}

\newcommand{\blocks}{\mathcal{B}}
\newcommand{\blockset}{\mathcal{B}}
\newcommand{\vbp}{v_{BP}}
\newcommand{\vbphat}{\hat{v}_{BP}}
\newcommand{\basefee}{r}

\title{Transaction Fee Mechanism Design with Active Block
  Producers\thanks{First version published July 6, 2023.}}
\author{Maryam Bahrani\thanks{a16z crypto. Email:
    \texttt{mbahrani@a16z.com}.} \and Pranav Garimidi\thanks{a16z
    crypto. Email: \texttt{pgarimidi@a16z.com}.} \and Tim
  Roughgarden\thanks{Columbia University \& a16z crypto. 
Author's research at Columbia
  University supported in part by NSF awards CCF-2006737
and CNS-2212745.
Email: \texttt{troughgarden@a16z.com}. 
}}

\begin{document}

\maketitle

\begin{abstract}
The incentive-compatibility properties of blockchain transaction fee
mechanisms have been investigated with 
{\em passive} block producers that are motivated purely by the net
rewards earned at the consensus layer. This paper introduces a model
of {\em active} block producers that have their own private valuations
for blocks (representing, for example, additional value derived from
the application layer).  The block producer surplus in our model can
be interpreted as one of the more common colloquial meanings of the
term ``MEV.''

The main results of this paper show that transaction fee mechanism
design is fundamentally more difficult with active block producers
than with passive ones: with active block producers, no non-trivial or
approximately welfare-maximizing transaction fee mechanism can be
incentive-compatible for both users and block producers.
These results can be interpreted as a mathematical justification for
the current interest in augmenting transaction fee mechanisms with
additional components such as order flow auctions, block producer
competition, trusted hardware, or cryptographic techniques.
\end{abstract}

\section{Introduction}

\subsection{Transaction Fee Mechanisms for Allocating Blockspace}

Blockchain protocols such as Bitcoin and Ethereum process transactions
submitted by users, with each transaction advancing the ``state'' of
the protocol (e.g., the set of Bitcoin UTXOs, or the state of the
Ethereum Virtual Machine).  Such protocols have finite processing
power, so when demand for transaction processing exceeds the available
supply, a strict subset of the submitted transactions must be chosen
for processing.  To encourage the selection of the ``most valuable''
transactions, the transactions chosen for processing are typically
charged a transaction fee. The component of a blockchain protocol
responsible for choosing the transactions to process and what to
charge for them is called its {\em transaction fee mechanism (TFM)}.

Previous academic work on TFM design (surveyed in Section~\ref{ss:rw})
has focused on the game-theoretic properties of different designs,
such as incentive-compatibility from the perspective of users
(ideally, with a user motivated to bid its true value for the
execution of its transaction), of block producers (ideally, with a
block producer motivated to select transactions to process as
suggested by the TFM), and of cartels of users and/or block producers.
Discussing incentive-compatibility requires defining utility
functions for the relevant participants.
In most previous works on TFM design (and in
this paper), users are modeled as having a private value for
transaction inclusion and a quasi-linear utility function (i.e., value
enjoyed minus price paid). In previous work---and, crucially, unlike
in this work---a block producer was modeled as {\em passive}, meaning
its utility function was the net reward earned (canonically, the
unburned portion of the transaction fees paid by users, possibly plus
a block reward).

While this model is a natural one for the initial investigation of the
basic properties of TFMs, it effectively assumes that block producers
are unaware of or unconcerned with the semantics of the transactions
that they process---that there is a clean separation between users
(who have value only for activity at the application layer) and block
producers (who, if passive, care only about payments received at the
consensus layer).

\subsection{MEV and Active Block Producers}\label{ss:mev}

It is now commonly accepted that, at least for blockchain protocols
that support a decentralized finance (``DeFi'') ecosystem, there are
unavoidable interactions between the consensus layer (block producers)
and the application layer (users), and specifically with block
producers deriving value from the application layer that depends on
which transactions they choose to process (and in which order).  For a
canonical example, consider a transaction that executes a trade on an
automated market maker (AMM), exchanging one type of token for another
(e.g., USDC for ETH).  The spot price of a typical AMM moves with
every trade, so by executing such a transaction, a block producer may
move the AMM's spot price out of line with the external market (e.g.,
on centralized exchanges (CEXs)), thereby opening up an arbitrage
opportunity (e.g., buying ETH on a CEX at the going market price and
then selling it on an AMM with a larger spot price). The block
producer is uniquely positioned to capture this arbitrage opportunity,
by executing its own ``backrunning'' transaction (i.e., a trade in the
opposite direction) immediately after the submitted trade transaction.

The first goal of this paper is to generalize the existing models of
TFM design in the minimal way that accommodates {\em active} block
producers, meaning block producers with a utility function that
depends on both the transactions in a block (and their order) and
the net fees earned. Specifically, in addition to the standard
private valuations for transaction inclusion possessed by users, the
block producer will have its own private valuation, which is an
abstract function of the block that it publishes. 
We then assume that a block producer acts to maximize its
{\em block
  producer surplus (BPS)}, meaning its private value for the published
block plus any addition profits (or losses) from fees (or burns).
In the interests of a simple but general model, we deliberately
avoid microfounding the private valuation function of a block producer
or committing to any specifics of the application layer. Our model
captures, in particular, canonical on-chain DeFi opportunities such as
arbitrage and liquidation opportunities, but a block producer's
valuation can reflect arbitrary preferences, perhaps derived also from 
off-chain activities (e.g., a bet with a friend) or subjective
considerations.

The extraction of application-layer value by block producers, in DeFi
and more generally, was first studied by Daian et al.~\cite{flashboys}
under the name ``MEV'' (for ``maximal extractable value,'' n\'ee
``miner extractable value'').  At this point, the term has transcended
any specific definition---in both the literature and popular
discourse, it is used, often informally, to refer to a number of
related but different concepts.  We argue that our definition of BPS
captures, in a precise way and in a concrete economic model, one of
the more common colloquial meanings of the term ``MEV.''  There are,
of course, many further interesting aspects of MEV, but these are
outside the scope of this paper.

In this paper, we treat a block producer as a single entity that
publishes a block based on the transactions that it is aware of.  This
would be an accurate model of block production, as carried out by miners in
proof-of-work protocols and validators in proof-of-stake protocols, up
until a few years ago.  More recently, especially in the Ethereum
ecosystem, block production has evolved into a more complex process,
typically involving ``searchers'' (who identify opportunities for
extraction from the application layer), ``builders'' (who assemble
such opportunities into a valid block), and ``proposers'' (who
participate directly in the blockchain protocol and make the final
choice of the published block).  One interpretation of a block
producer in our model is as a vertically integrated searcher, builder,
and proposer.  In future work, we will consider more fine-grained
models of the block production process. As we'll see in this paper,
there is much to say about transaction fee mechanism design with
active block producers already when block production is carried out by
a single entity.

\subsection{Overview of Results}\label{ss:overview}

Our starting point is the model for transaction fee mechanism design
defined in~\cite{roughgarden2021transaction}.  In this model, each user has a private
valuation for the inclusion of a transaction in a block, and submits a
bid along with its transaction.  As in~\cite{roughgarden2021transaction}, we consider
TFMs that choose the included transactions and payments based solely
on the bids 
of the pending transactions (as opposed to, say,
based also on something derived from the semantics of those
transactions).  A block producer
publishes any block that it wants, subject to feasibility (e.g., with
the total size of the included transactions respecting some maximum
block size).
A TFM is said to be {\em dominant-strategy
  incentive-compatible (DSIC)} if every user has a
dominant (i.e., always-optimal) bidding strategy.
The DSIC property is often associated with a good
``user experience (UX),'' in the sense that each user has an obvious
optimal bid.
In~\cite{roughgarden2021transaction}, a TFM was said to be {\em incentive-compatible for
  myopic miners (MMIC)} if it expects a block producer to publish a
block that maximizes the net fees earned (at the consensus layer).
Here, we introduce an analogous definition that accommodates active
block producers: We call a TFM {\em incentive-compatible for block
  producers (BPIC)} if it expects a block producer to publish a block
that maximizes its private valuation 
plus the net fees
earned.  An ideal TFM would satisfy, among other properties, both DSIC
and BPIC.

This paper shows that TFM design is fundamentally more difficult with
active block producers than with passive ones.  Our first result
(Theorem~\ref{t:main}) is a proof that, with active block producers,
{\em no} non-trivial TFM satisfies both DSIC and BPIC, where
``non-trivial'' means that users must at least in some cases pay a
nonzero amount for transaction inclusion. (In contrast, with passive
block producers, the ``tipless mechanism'' suggested
in~\cite{roughgarden2021transaction} is non-trivial and satisfies both
DSIC and BPIC; the same is true of the EIP-1559 mechanism of
Buterin et al.~\cite{buterin2019eip}, provided the mechanism's base fee is not
excessively low~\cite{roughgarden2021transaction}.) In
particular, the EIP-1559 and tipless mechanisms fail to satisfy DSIC
and BPIC when block producers can be active. Intuitively, for these
mechanisms, a user might be motivated to underbid in the hopes of an
effective subsidy by the block producer (who may include the
transaction anyways, if it derives outside value from it).  In fact,
we show in Theorem~\ref{t:dsic_approx} that, in a quantitative sense,
the loss in DSIC of these mechanisms scales precisely with the value
extractable by the block producer from users' transactions.
Theorem~\ref{t:main} shows that this breakdown in
incentive-compatibility is not a failure of these TFMs per se, but
rather stems from a fundamental obstacle to TFM design with active
block producers.

The second main result of the paper (Theorem~\ref{t:welfare}) proves
that TFMs that do not charge non-zero transaction fees---and in
particular (by Theorem~\ref{t:main}), TFMs that are both DSIC and
BPIC---cannot provide any meaningful welfare-maximization guarantees.
Intuitively, the issue is the lack of alignment between the
preferences of users and of the block producer: If a block producer
earns no transaction fees from any block, it might choose a block with
non-zero private value but only very low-value transactions over one
with no private value but very high-value transactions.

Impossibility results like Theorems~\ref{t:main} and~\ref{t:welfare} are
meant to guide rather than discourage further work on the problem, by
highlighting the paths forward along which positive progress might be
made.  In fact, these results can be interpreted as a mathematical
justification for the community's current interest in augmenting
transaction fee mechanisms with additional components such as order
flow auctions (e.g., \cite{flashbots1}), block producer competition
(e.g., \cite{dom_mev_burn}), trusted hardware (e.g.,
\cite{flashbots2}), or cryptographic techniques (e.g.,
\cite{encrypted_mempools}).

\subsection{Related Work}\label{ss:rw}

\paragraph{Defining MEV.}
Daian et al.~\cite{flashboys} introduced the notion of miner/maximal
extractable value. They defined MEV as the value that miners or
validators could obtain by manipulating the transactions in a
block. Since this work there have been many follow-up works attempting
to formalize MEV and analyze its effects in both theory and
practice. Attempts to give exact theoretical characterizations of MEV
appear in
\cite{towardsMEVFormalization,crossChainMEV,Bartoletti,BabelDKJ21}. Broadly,
these works define MEV by defining sets of valid transaction sequences
and allowing the block producer to maximize their value over these
sequences.  These definitions are very general, but in exchange have
to this point proved analytically intractable.
Several empirical papers study the impact and magnitude of
MEV 
using heuristics applied to on-chain data
\cite{QinZG22,QinZLG21,TorresCS21}. Another line of work
\cite{tarunCFMM,HeimbachW22,BartolettiCL22} studies MEV in specific
contexts, such as for arbitrage in AMMs, in which case it is possible
to characterize how much MEV can be realized from certain
transactions. In particular, Kulkarni et al.~\cite{tarunCFMM} give
formal statements on how, under different AMM designs, MEV affects the
social welfare of the overall system.

\paragraph{General TFM literature.}
The model in this paper is closest to the one used by
Roughgarden~\cite{roughgarden2021transaction} to analyze (with passive
block producers) the economic properties of the EIP-1559
mechanism~\cite{buterin2019eip}, the TFM used currently in the
Ethereum blockchain.
Precursors to that work (also with passive block producers) include
studies of a ``monopolistic price'' transaction fee
mechanism~\cite{LSZ19,Y18} (also considered recently by
Nisan~\cite{N23}), and work of Basu et al.~\cite{beos} that proposed a
more sophisticated variant of that mechanism.  There have also been
several follow-up works to~\cite{roughgarden2021transaction} that use
similar models (again, with passive block producers).  Chung and
Shi~\cite{chung2023foundations} prove impossibility results showing
that the incentive-compatible guarantees of the EIP-1559 mechanism are
in some sense optimal.  There have also been attempts to circumvent
this impossibility result by relaxing the notion of incentive
compatibility~\cite{chung2023foundations,gafni2022greedy}, using
cryptography~\cite{shi2022can}, considering a Bayesian
setting~\cite{zhao2022bayesian}, or mixtures of these
ideas~\cite{wu2023maximizing}. 
Other recent works~\cite{ferreira2021dynamic,monnot_aft} 
study the dynamics of the base fee in the 
EIP-1559 mechanism. 

A more distantly related line of work involves mechanism design in the
presence of strategic auctioneers.
Akbarpour~and~Li~\cite{akbarpour-credible-journal} introduce the
notion of \emph{credible} mechanisms, where any profitable deviations
by the auctioneer can be detected by at least one user.  While similar
in spirit to the concept of BPIC introduced here (and the special case
of MMIC introduced in~\cite{roughgarden2021transaction}), there are
several important differences.  For example, the theory of credible
mechanisms assumes fully private communication between bidders and the
auctioneer and no communication among bidders, whereas TFM bids are
commonly collected from a public mempool. There is also a line of follow-up
work that takes advantage of cryptographic primitives to build credible
auctions on the
blockchain~\cite{ferreira2020credible,essaidi2022credible,chitra2023credible,ferreira2022credible}.

\section{The Model}\label{s:model}

This section defines transaction fee mechanisms, the relevant players
and their objectives, and notions of incentive-compatibility for users
and for block producers. Our model can be viewed as the salient parts
of the one in~\cite{roughgarden2021transaction}, augmented with the
necessary ingredients to discuss active block producers.

\subsection{The Players and Their Objectives}\label{ss:utilities}

\paragraph{Users.}
Users submit {\em transactions} to the blockchain protocol. The
execution of a transaction updates the state of the protocol (e.g.,
users' account balances). The rules of the protocol specify whether a
given transaction is {\em valid} (e.g., whether it is accompanied by
the required cryptographic signatures). From now on, we assume that
all transactions under consideration are valid.

We assume that
each user submits a
single transaction~$t$ and has a nonnegative {\em valuation}~$v_t$,
denominated in a base currency like USD or ETH, for its execution in
the next block. This valuation is {\em private}, in the sense that it
is initially unknown to all other parties.  We assume that the utility
function of each user---the function that the user acts to
maximize---is quasi-linear, meaning that its utility is either~0 (if
its transaction is not included in the next block) or~$v_t-p$ (if
its transaction is included and it must pay a fee of~$p$). The primary
focus of this paper is on impossibility results, and this restrictive
combination of myopic users, inclusion valuations, and quasi-linear
utility only makes such results stronger.

\vspace{-.1in}

\paragraph{Blocks.} 
A {\em block} is a finite ordered list of transactions. A {\em feasible
  block} is a block that respects any additional constraints imposed
by the protocol. 
For example, if transactions have sizes (e.g., the ``gas limit'' in
Ethereum) and the protocol specifies a maximum block size, then
feasible blocks might be defined as those that comprise only valid
transactions and also respect the block size limit.

\vspace{-.1in}

\paragraph{Block producers (BPs).} We consider blockchain protocols
for which the contents of each block are selected by a single entity,
which we call the {\em block producer (BP)}. We focus on the
decision-making of the BP that has been chosen at a particular moment
in time (perhaps using a proof-of-work or proof-of-stake-based
lottery) to produce the next block. We assume that whatever block
the BP chooses is in fact published, with all the included
transactions finalized and executed.

A BP chooses a block~$B$ from some abstract non-empty set~$\blocks$ of
feasible blocks, called its {\em blockset}.  For example, the
set~$\mathcal{B}$ might consist of all the feasible blocks that
comprise only transactions that the BP knows about (perhaps from a
public mempool, or perhaps from private communications) along with
transactions that the BP is in a position to create itself (e.g., a
backrunning transaction).  As with users, we model the preferences of
a BP with a quasi-linear utility function, meaning the difference
between its private value for a block (again, denominated in a base
currency like USD or ETH) minus the (possibly negative) payment that
it must make.  Unlike with users, to avoid modeling any details of why
a BP might value a block (e.g., due to the extraction of value from
the application layer), we allow a BP to have essentially arbitrary
preferences over blocks. More formally, we assume that a BP has a
private valuation that is an arbitrary (real-valued) function~$\vbp$
over blocks, and that the BP acts to maximize its {\em block producer
  surplus (BPS)}:
\[
\underbrace{\vbp(B) + \text{net fees earned}}_{\text{block producer
    surplus (BPS)}}.
\]
For example, in a mechanism like EIP-1559 (see Examples~\ref{ex:1559}
and~\ref{ex:15592} below), ``net fees''
would correspond to the unburned portion of the transaction fees paid
by users.\footnote{We assume that the marginal cost to a BP of
  including a transaction in a block (e.g., due to an increase in orphaning
  probability in a longest-chain protocol) is~0.  Non-zero marginal
  costs can be easily incorporated into our model,
  following~\cite{roughgarden2021transaction}.}

Previous work on TFMs considered the special case of {\em passive} BPs,
meaning BPs with a block-independent valuation (perhaps everywhere
equal to~0, or to the value of a fixed block reward).  This special
case effectively assumes that a BP has no out-of-protocol value for
any transaction. For a simple example of an {\em active} BP, consider
an {\em additive} valuation, meaning a function $\vbp$ of the form
\[
\underbrace{\vbp(B) := \sum_{t \in B} \mu_t}_{\text{additive valuation}},
\]
where ``$t \in B$'' sums over the transactions~$t$ in the block~$B$
and where the $\mu_t$'s are transaction-specific constants.
Intuitively, an additive valuation corresponds to a BP that extracts
value from each transaction independently, without regard to possible
interactions between the transactions in its block. For example, if
each transaction corresponds to a trade in a distinct AMM, the BP's
valuation might reasonably be modeled as additive.  Alternatively, if
the BP is focused on capturing a particular liquidation opportunity in
a lending platform, its valuation might reasonably be modeled as {\em
  single-minded}, meaning that for every block~$B$, $\vbp(B)$ is
either equal to~0 (if the BP does not capture the liquidation
opportunity in~$B$) or to some fixed constant~$\mu$ (otherwise,
where~$\mu$ is the value of the opportunity):
\[
\underbrace{
\vbp(B) = \left\{ 
\begin{array}{cl}
\mu & \text{if $B \in \mathcal{S}$}\\
0 & \text{otherwise}
\end{array},
\right.
}_{\text{single-minded valuation}}
\]
where~$\mathcal{S} \sse \blocks$ is a subset of feasible blocks
and~$\mu$ is a constant.

As noted in Section~\ref{ss:mev}, in this paper we primarily regard a
BP as a single entity.  If instead blocks are produced by a collection
of parties, $\vbp(B)$ should be interpreted as the combined value of
all these parties for a block. Our assumption that a BP will choose a
block to maximize its BPS translates to the assumption that the
collection of parties is colluding to maximize their total surplus.

\vspace{-.1in}

\paragraph{Holders.} The final category of participants, which are
non-strategic in our model but relevant for our definition of welfare
in Section~\ref{ss:welfare}, are the holders of the blockchain
protocol's native currency.  As we'll see in Section~\ref{ss:tfm},
TFMs are in a position to mint or burn this currency, which
corresponds to inflation or deflation, respectively.  We treat TFM
mints and burns as transfers from and to, respectively, the existing
holders of this currency. For example, when a TFM burns some of the
transaction fees paid by the users that created the included
transactions (as in the EIP-1559 mechanism), we interpret the outcome
as a transfer of funds from the TFM's users to the currency's holders.
Formally, we define the collective utility function of currency
holders to be the net amount of currency burned by a TFM.

\subsection{Welfare}\label{ss:welfare}

According to the principle of welfare-maximization, a scarce resource
like blockspace should be allocated to maximize the total utility of
all the ``relevant participants,'' which in our case includes the
users, the BP, and the currency holders.  Because all parties have
quasi-linear utility functions and all TFM transfers will be between
members of this group (from users to the BP, from the BP to holders,
etc.), the welfare of a block is simply the sum of the user and BP
valuations for it:
\begin{equation}\label{eq:welfare}
\underbrace{W(B) := \vbp(B) + \sum_{t \in B} v_t}_{\text{welfare of $B$}}.
\end{equation}
Holders are assumed to be passive and thus have no
valuations to contribute to the sum.\footnote{We stress that the
  welfare of a block~\eqref{eq:welfare} measures
the ``size of the pie'' and says nothing about 
how this welfare might be split between users, the BP, and
holders (i.e., about the size of each slice). Distributional
considerations are important, of course, but they are outside the
scope of this paper.}

\subsection{Transaction Fee Mechanisms}\label{ss:tfm}

The outcome of a transaction fee mechanism is a block to publish and a
set of transfers (user payments, burns, etc.) that will be made upon
the block's publication. 
In line with the preceding literature on TFMs
and the currently deployed TFM designs, we assume that each user that
creates a transaction~$t$
submits along with it a nonnegative
{\em bid}~$\bidt$ (i.e., willingness to pay),
and that a TFM bases its transfers on the set of
available transactions and the corresponding bids.  (The BP submits
nothing to the TFM.)  A TFM is defined primarily by its {\em payment}
and {\em burning} rules, which specify the fees paid by users and the
burned funds implicitly received by holders (with the BP pocketing the
difference).

\vspace{-.1in}

\paragraph{Payment and burning rules.}
The payment rule specifies the payments made by users in
exchange for transaction inclusion.
\begin{definition}[Payment Rule]\label{d:payment}
A {\em payment rule} is a function~$\prices$ that specifies a
nonnegative payment
$\pricet(B,\bids)$ for each transaction~$t \in B$ in a block~$B$,
given the bids~$\bids$ of all known transactions.
\end{definition}
The value of $\pricet(B,\bids)$ indicates the
payment from the creator of an included transaction~$t \in B$ to the
BP that published that block.  
(Or, if the rule is randomized, the expected payment.\footnote{We
  assume that users and BPs are risk-neutral when interacting with a
  randomized TFM.})
We consider only {\em individually
  rational} payment rules, meaning that $\pricet(B,\bids) \le \bidt$
for every included transaction $t \in B$.
We can interpret $\pricet(B,\bids)$ as 0 whenever $t \notin B$.
Finally, we assume that every creator of an included transaction has
the funds available to pay its full bid, if necessary (otherwise, the
block~$B$ should be considered infeasible).

The burning rule specifies how much money must be burned by a BP along
with the publication of a given block.\footnote{This differs
  superficially from the formalism in~\cite{roughgarden2021transaction}, in which a
  burning rule
  specifies per-transaction (rather than per-block) transfers from users
  (rather than the BP) to currency holders.
The payment rule here can be interpreted as the sum of the payment and
burning rules in~\cite{roughgarden2021transaction}, and the per-block burning rule here
can be interpreted as the sum of the burns of a block's transactions in~\cite{roughgarden2021transaction}.}
\begin{definition}[Burning Rule]\label{d:burning}
A {\em burning rule} is a function~$\burn$ that specifies a
nonnegative burn~$\burn(B,\bids)$ for a block~$B$, given the
bids~$\bids$ of all known transactions.
\end{definition}
The value of $\burnt(B,\bids)$ indicates the amount of money burned
(i.e., paid to currency holders) by the BP upon publication of the
block~$B$. (Or, if the rule is randomized, the expected
amount.)\footnote{An alternative to money-burning that has similar
  game-theoretic and welfare properties is to transfer~$\burn(B,\bids)$ to
  entities other than the BP, such as a foundation or the producers of
  future blocks.}  We assume that, after receiving users' payments for
the block, the BP has sufficient funds to pay the burn required of the
block that it publishes (otherwise, the block~$B$ should be considered
infeasible).

The most commonly used TFMs restrict their payment and burning rules
to depend only on the bids of the included transactions, as opposed to
the bids of all known (included and excluded) transactions.  The
positive results in this paper (Theorem~\ref{t:dsic_approx}
and Remark~\ref{rem:welfare_approx}) use only TFMs of this restricted type. The
negative results (Theorems~\ref{t:main} and~\ref{t:welfare}) apply
even to TFMs with payment and burning rules of the more general form
in Definitions~\ref{d:payment} and~\ref{d:burning}.\footnote{Chung and
  Shi~\cite{chung2023foundations} show that, in some scenarios, TFMs with general
  payment and burning rules can satisfy properties that cannot be
  achieved with payment and burning rules that depend only on the bids
  of the transactions chosen for execution.}

We stress that the payment and burning rules of a TFM are hard-wired
into a blockchain protocol as part of its code. This is why their
arguments---the transactions chosen for execution and their bids, and
perhaps (as in~\cite{chung2023foundations}) the bids of some
additional, not-to-be-executed transactions---must be publicly
recorded as part of the blockchain's history. (E.g., late arrivals
should be able to reconstruct users' balances, including any payments
dictated by a TFM, from this history.)  A BP cannot manipulate the
payment and burning rules of a TFM, except inasmuch as it can choose
which block~$B \in \blocks$ to publish.

\begin{example}[First-Price Auction (FPA)]\label{ex:fpa}
A {\em first-price auction (FPA)}, as used in the Bitcoin protocol and
pre-2021 Ethereum, simply charges users their bids. That is, the
burning rule~$q$ is identically zero, and the payment rule is
$\pricet(B,\bids) = \bidt$ for all $t \in B$.
\end{example}

\begin{example}[EIP-1559]\label{ex:1559}
The EIP-1559 mechanism~\cite{buterin2019eip} assumes that every
transaction~$t$ has a 
publicly known {\em size}~$s_t$ (e.g., the gas limit of an Ethereum
transaction). The mechanism is parameterized by a ``base fee''
$\basefee$, which for each transaction~$t$ (with size~$s_t$)
defines a reserve price of $\basefee \cdot s_t$. 
Like an FPA, this mechanism charges each user its bid:
$\pricet(B,\bids) = \bidt$ for all $t \in B$.
Unlike an FPA, the portion of this revenue generated by the base fee
goes to holders rather than the BP.  That is, the mechanism's burning
rule is~$\burn(B,\bids) = \sum_{t \in B} \basefee \cdot s_t$.
(We allow a BP to include transactions with $\bidt < \basefee \cdot s_t$,
but the BP must still burn the full amount $\basefee \cdot
s_t$; see also Remark~\ref{rem:constrained}.)
\end{example}

\vspace{-.1in}

\paragraph{Allocation rules.}
In our model, a BP has unilateral control over the block that it
chooses to publish.  Thus, a TFM's allocation rule---which specifies
the block that should be published, given all of the relevant
information---can only be viewed as a recommendation to a BP.  Because
the (suggested) allocation rule is carried out by the BP and not
by the TFM directly, it can sensibly depend on arguments not known to
the TFM (but known to the BP), specifically the BP's valuation~$\vbp$
and blockset~$\blocks$.
\begin{definition}[Allocation Rule]\label{d:allocation}
An {\em allocation rule} is a function~$\allocs$ that specifies a
block\\ $\allocs(\bids,\vbp,\blocks) \in \blocks$, given the bids~$\bids$ of all
known transactions, the BP valuation~$\vbp$, and the BP blockset~$\blocks$.
\end{definition}
An allocation rule~$\allocs$ induces per-transaction allocation
rules with, for a transaction~$t$, $\alloct(\bids,\vbp,\blocks)=1$ if
$t \in \allocs(\bids,\vbp,\blocks)$ and~0 otherwise.

\begin{definition}[Transaction Fee Mechanism (TFM)]\label{d:tfm}
A {\em transaction fee mechanism (TFM)} is a triple $\tfm$ in which
$\allocs$ is a (suggested) allocation rule, $\prices$ is a payment rule,
and $\burn$ is a burning rule.
\end{definition}
A TFM is defined relative to a specific block publishing opportunity.
A blockchain protocol is free to use different TFMs for different
blocks (e.g., with different base fees), perhaps informed by the
blockchain's past history.

\vspace{-.1in}

\paragraph{Utility functions and BPS revisited.}
With Definitions~\ref{d:payment}--\ref{d:tfm} in place, we can express
more precisely the strategy spaces and utility functions introduced in
Section~\ref{ss:utilities}.  We begin with an expression for the
utility of a user (as a function of its bid) for a TFM's outcome, under
the assumption that the BP always chooses the block suggested by the
TFM's allocation rule.
\begin{definition}[User Utility Function]\label{d:userutil}
For a TFM $\tfm$, BP valuation~$\vbp$, BP blockset~$\blocks$, and bids
$\bidsmt$ of other transactions,
the utility of the originator of a transaction~$t$
with valuation~$\valt$ and bid $\bidt$ is
\begin{equation}\label{eq:userutil}
u_t(\bidt) :=
\valt \cdot \alloct((\bidt,\bidsmt),\vbp,\blocks) - \pricet(B,(\bidt,\bidsmt)),
\end{equation}
where~$B:=\alloct((\bidt,\bidsmt),\vbp,\blocks)$.
\end{definition}
In~\eqref{eq:userutil}, we highlight the dependence of the utility
function on the argument that is directly under a user's control, the
bid~$\bidt$ submitted with its transaction.  

The BP's utility function, the block producer surplus, is then:
\begin{definition}[Block Producer Surplus (BPS)]\label{d:bputil}
For a TFM $\tfm$, BP valuation~$\vbp$, BP blockset $\blocks$,
and transaction bids $\bids$,
the block producer surplus of a BP that chooses the block $B \in \blocks$ is
\begin{equation}\label{eq:bputil}
u_{BP}(B) := \vbp(B) + \sum_{t \in B} \pricet(B,\bids)
 - \burn(B,\bids).
\end{equation}
\end{definition}
In~\eqref{eq:bputil}, we highlight the dependence of the BP's utility
function on the argument that is under its direct control, its
choice of a block.  The BP's utility depends on the payment and burning
rules of the TFM, but not on its allocation rule (which the BP is free
to ignore, if desired).

Finally, the collective utility function of (passive) currency holders
for a block~$B$ with transaction bids~$\bids$ is~$\burn(B,\bids)$, the
amount of currency burned by the BP.  (As promised, for a block~$B$,
no matter what the bids and the TFM, the sum of the utilities of
users, the BP, and holders is exactly the welfare defined
in~\eqref{eq:welfare}.)

\subsection{Incentive-Compatible TFMs}\label{ss:ic}

In this paper, we focus on two incentive-compatibility notions for
TFMs---which, as we'll see, are already largely incompatible---one
for users and one for block producers.  We begin with the latter.

\vspace{-.1in}

\paragraph{BPIC TFMs.}
We assume that a BP will choose a block to maximize its utility
function, the BPS (Definition~\ref{d:bputil}).  The defining
equation~\eqref{eq:bputil} shows that, once the payment and burning
rules of a TFM are fixed, a BP's valuation and blockset dictate the
unique (up to tie-breaking) BPS-maximizing block for each bid vector.
We call an
allocation rule {\em consonant} if, given the payment and burning
rules, it instructs a BP to always choose such a block (breaking ties
in an arbitrary but consistent fashion).
\begin{definition}[Consonant Allocation Rule]\label{d:consonant}
An allocation rule $\allocs$ is {\em consonant} with the payment and
burning rules~$\prices$ and~$\burn$ if:
\begin{itemize}

\item [(a)]  for every BP valuation~$\vbp$
and blockset~$\blocks$, and for every choice of transaction
bids~$\bids$,
\[
\underbrace{\allocs(\bids,\vbp,\blocks)}_{\text{recommended block}}
\in \,\,\, \underbrace{\argmax_{B \in \blocks}\left\{ \vbp(B) + \sum_{t \in B}
    \pricet(B,\bids)
 - \burn(B,\bids) \right\}}_{\text{BPS-maximizing block}};
\]
\item [(b)] for some fixed total ordering on the blocks of~$\blocks$,
  the rule breaks ties between BPS-maximizing blocks according to this
  ordering.

\end{itemize}
\end{definition}

{\em BPIC} TFMs are then precisely those that always instruct a BP to
choose a BPS-maximizing block.
\begin{definition}[Incentive-Compatibility for Block Producers (BPIC)]\label{d:bpic}
A TFM $\tfm$ is {\em incentive-compatible for block producers (BPIC)}
if its allocation rule~$\allocs$ is consonant with its payment
rule~$\prices$ and burning rule~$\burn$.
\end{definition}

\vspace{-.1in}

\paragraph{DSIC TFMs.}
Dominant-strategy incentive-compatibility (DSIC) is one way to
formalize the idea of ``good user experience (UX)'' for TFMs.  The
condition asserts that every user has an ``obviously optimal'' bid,
meaning a bid that, provided the BP follows the TFM's allocation rule,
is guaranteed to maximize the user's utility (no matter what other
users might be bidding).  In the next definition, by a {\em bidding
  strategy}, we mean a function$~\sigma$ that maps a valuation to a
recommended bid for a user with that valuation.
\begin{definition}[Dominant-Strategy Incentive-Compatibility
  (DSIC)]\label{d:dsic}
A TFM $\tfm$ is {\em dominant-strategy incentive-compatible (DSIC)} 
if there is a bidding strategy~$\sigma$ such that, for every BP
valuation~$\vbp$ and blockset~$\blocks$, every user~$i$ with
transaction~$t$, every valuation~$\valt$ for~$i$, and every choice of
other users' bids~$\bidsmt$,
\begin{equation}\label{eq:dsic}
\underbrace{\sigma(\valt)}_{\text{recommended bid}} \in
\underbrace{\argmax_{\bidt} \{  u_t(\bidt)
  \}}_{\text{utility-maximizing bid}},
\end{equation}
where~$u_t$ is defined as in~\eqref{eq:userutil}.  
\end{definition}
That is, bidding according to the recommendation of the bidding
strategy~$\sigma$ is guaranteed to maximize a user's
utility.\footnote{The term ``DSIC'' is often used to refer
  specifically to mechanisms that satisfy the condition in
  Definition~\ref{d:dsic} with the truthful bidding strategy,
  $\sigma(\valt) = \valt$.  Any mechanism that is DSIC in the sense of
  Definition~\ref{d:dsic} can be transformed into one in which
  truthful bidding is a dominant strategy, simply by enclosing the
  mechanism in an outer wrapper that accepts truthful bids, applies
  the assumed bidding strategy~$\sigma$ to each, and passes on the
  results to the given DSIC mechanism.  (This trick is known as the
  ``Revelation Principle''; see
  e.g.~\cite{20lectures}.)\label{foot:rp}}
This is a strong property: a bidding strategy can depend only on what
a user knows (i.e., its private valuation), while the right-hand side
of~\eqref{eq:dsic} implicitly depends (through~\eqref{eq:userutil})
also on the bids of the other users and the BP's preferences.

\begin{example}[First-Price Auctions Revisited]\label{ex:fpa2}
The optimal bid for a user in an FPA (Example~\ref{ex:fpa}) generally
depends on other users' bids, so FPAs are not DSIC in the sense of
Definition~\ref{d:dsic}. Are FPAs BPIC? The answer depends on how the
allocation rule $\allocs$ of an FPA is defined. The usual
definition of an FPA~\cite{roughgarden2021transaction} asserts that a
BP should choose the block
that maximizes the BP's revenue from transaction fees---the block that
maximizes the sum of the bids of the included transactions.  This
allocation rule is consonant with the FPA's payment and burning rules
if the BP is passive, but not if the BP is active.  With an active BP,
the consonant allocation rule would instruct a BP to maximize its
BPS---the sum of the bids of the included transactions {\em plus} any
private value that the BP has for the block.  If an FPA's allocation
rule is redefined in this way to ensure consonance, then the FPA becomes
BPIC.
\end{example}

\begin{example}[EIP-1559 Revisited]\label{ex:15592}
Following~\cite{roughgarden2021transaction}, in the EIP-1559 mechanism
(Example~\ref{ex:1559}), call the base fee $\basefee$ {\em excessively
  low} if the BP cannot fit all the transactions~$t$ satisfying $\bidt
\ge \basefee \cdot s_t$ into a single (feasible) block. (Recall
that~$s_t$ denotes the publicly known ``size'' of a transaction~$t$.)
When the base fee is not excessively low,
the standard allocation rule for the EIP-1559 mechanism
instructs the BP to include all transactions~$t$
for which $\bidt \ge \basefee \cdot s_t$ (and to leave out
any transactions~$t$ with $\bidt < \basefee \cdot s_t$).
With a passive BP, this allocation rule is consonant with
the payment and burning rules of the mechanism: In this case, including a
transaction~$t$ in the block contributes precisely $\bidt - \basefee
\cdot s_t$ to the BPS, so a passive BP is motivated to
include all and only the
transactions for which this expression is nonnegative.  With this
allocation rule (and a base fee that is not excessively low), the TFM
is also DSIC, with the bidding
strategy~$\sigma$ defined by $\sigma(\valt) = \min \{ \valt, \basefee
\cdot s_t \}$.

With an active BP, however, the usual allocation rule above is no
longer consonant with the payment and burning rules of the mechanism,
even when the base fee is not excessively low:
A BP might be motivated to include a transaction~$t$ with $\bidt <
\basefee \cdot s_t$, if the deficit can be compensated for with the
BP's own private value for including the transaction.  Thus, this
version of the EIP-1559 mechanism is not BPIC. The mechanism's allocation
rule can be redefined to restore consonance, by instructing the BP to
choose the block that maximizes its BPS (rather than its revenue), but
this robs the mechanism of its DSIC property: Intuitively, without
knowing the BP's valuation, a user cannot know whether to underbid
(below its reserve price) to take advantage of a BP 
that might be willing to subsidize the difference.
\end{example}
The main result of this paper (Theorem~\ref{t:main}) shows that the
whack-a-mole between the DSIC and BPIC properties in
Example~\ref{ex:15592} is not particular to the EIP-1559 mechanism:
When BPs are active, {\em no} TFM that charges non-zero user fees can
be both DSIC and BPIC.

Our final example shows that, with a passive BP, the
DSIC and BPIC properties can be achieved simultaneously even without
the assumption in Example~\ref{ex:15592} about the accuracy of a base
fee.
\begin{example}[Tipless Mechanism]\label{ex:tipless}
The {\em tipless mechanism}~\cite{roughgarden2021transaction} is a variation on the
EIP-1559 mechanism that removes the user-specified ``tips.''
Formally, the burning rule is the same as in Example~\ref{ex:1559}
(i.e., $\burn(B,\bids) = \sum_{t \in B} \basefee \cdot s_t$), while
the payment rule changes from $\pricet(B,\bids) = \bidt$ 
to $\pricet(B,\bids) = \min\{\bidt, \basefee \cdot
s_t\}$ for $t \in B$.  The mechanism's allocation rule instructs the
BP to include
only transactions~$t$ satisfying $\bidt \ge \basefee \cdot s_t$ and,
subject to this constraint and block feasibility, to maximize the
total size of the included transactions.  (Ties are broken according
to some fixed ordering over feasible blocks.)
The contribution of an included transaction to a BP's revenue is
either~0 (if $\bidt \ge \basefee \cdot s_t$) or negative (otherwise).
This implies that a passive BP cannot improve its BPS by deviating
from the allocation rule's recommendation.  
This TFM is also DSIC, under the same bidding strategy used in
Example~\ref{ex:15592} or, alternatively, under the truthful bidding
strategy.
\end{example}

\vspace{-.1in}

\paragraph{Off-chain agreements.}
For completeness, we briefly mention a third incentive-compatibility
notion, which concerns cartels that include a BP and one or more
users.  Such cartels can in some cases coordinate off-chain to
manipulate the intended behavior of a TFM.  For example, one of the
primary reasons that the EIP-1559 mechanism burns its base fee
revenue is resilience to coordination of this type. (If that
revenue was instead passed on to the BP, low-value users could
collude with the BP to evade the base fee, by overbidding on-chain to
clear the base fee while accepting a rebate from the BP off-chain.)
Informally, a TFM is {\em OCA-proof} if it is robust to collusion of
this type. (``OCA'' stands for ``off-chain agreement'';
see~\cite{roughgarden2021transaction} for the precise definition.)
OCA-proofness shaped the design of the EIP-1559 mechanism, and related
notions are fundamental to the TFM impossibility results (with passive
BPs) in~\cite{chung2023foundations}.\footnote{For example, one way to
  interpret the difference between the EIP-1559 mechanism
  (Example~\ref{ex:15592}) and the tipless mechanism
  (Example~\ref{ex:tipless}) is that, when the base fee is excessively
  low, the former mechanism gives up on DSIC (but retains OCA-proofness)
  while the latter gives up on OCA-proofness (but remains DSIC).}
OCA-proofness plays almost no role in this paper, because our
impossibility results (Theorems~\ref{t:main} and~\ref{t:welfare})
apply already to mechanisms that are merely DSIC and BPIC (and not
necessarily OCA-proof).

\begin{remark}[OCAs and the Two Versions of the EIP-1559 and Tipless
  Mechanisms]\label{rem:constrained}
In the versions of the EIP-1559 and tipless mechanisms described in
Examples~\ref{ex:15592} and~\ref{ex:tipless}, a BP is free to include
in a block any transaction it wants, whether or not the
bid~$\bidt$ submitted with the transaction is high enough to cover the
required burn $\basefee \cdot s_t$.  An alternative design would
change the definition of block feasibility so that such transactions
are ineligible for inclusion.\footnote{This would require a slight modification
  to our formalism in Section~\ref{ss:utilities}, with the feasibility
  of a block now depending on the bids attached to transactions,
  rather than solely on the transactions themselves.}
There is effectively no difference between the two
designs when BPs are passive: A rational such BP would never include a
transaction with $\bidt < \basefee \cdot s_t$, even were it free to do
so.  An active BP, however, will be motivated to include such a
transaction if it has a sufficiently high private value for it.

Off-chain agreements render these second versions of the EIP-1559 and
tipless mechanisms equivalent to those described in
Examples~\ref{ex:15592} and~\ref{ex:tipless}, even with active BPs.
The reason is similar to the reason why base fee revenue must be
withheld from a BP: If users collude with a BP, they can always bid
$\basefee \cdot s_t$ on-chain to ensure inclusion eligibility while
accepting an off-chain rebate of $\basefee \cdot s_t - \bidt$ from the
BP.
\end{remark}

\section{An Impossibility Result for DSIC and BPIC Mechanisms}

\subsection{Can DSIC and BPIC Be Achieved Simultaneously?}

The DSIC property (Definition~\ref{d:dsic}) encodes the idea of a
transaction fee mechanism with ``good UX,'' meaning that bidding is
straightforward for users.  Given the unilateral power of BPs in
typical blockchain protocols, the BPIC property
(Definition~\ref{d:bpic}) would seem necessary, absent any additional
assumptions, to have any faith that a TFM will be carried out by BPs
as intended.  One can imagine a long wish list of properties that we'd
like a TFM to satisfy; can we at least achieve these two?

The tipless mechanism (Example~\ref{ex:tipless}) is an example of a
TFM that is DSIC and BPIC in the special case of passive BPs. This TFM
is also ``non-trivial,'' in the sense that users generally pay for the
privilege of transaction inclusion.  With active BPs, meanwhile, the
DSIC and BPIC properties can technically be achieved simultaneously by
the following ``trivial'' TFM: the payment rule~$\prices$ and burning
rule~$\burn$ are identically zero, and the allocation rule $\allocs$
instructs the BP to choose the feasible block that maximizes its
private value (breaking ties in a bid-independent way).  This TFM is
BPIC by construction, and it is DSIC because a user has no control
over whether it is included in the chosen block (it's either in the
BP's favorite block or it's not) or its payment (which is always~0).

Thus, the refined version of the key question is:
\begin{quote}
Does there exist a non-trivial TFM that is DSIC and BPIC with
  active BPs?
\end{quote}

\subsection{Only Trivial Mechanisms Can Be DSIC and BPIC}

The main result of this paper is a negative answer to the preceding
question.  By the {\em range} of a bidding strategy~$\sigma$, we mean
the set of bid vectors realized by nonnegative valuations: $\{
\sigma(\vals) \,:\, \vals \ge 0 \}$, where~$\sigma(\vals)$ denotes the
componentwise application of~$\sigma$.
\begin{theorem}[Main Impossibility Result]\label{t:main}
If the TFM $\tfm$ is DSIC with bidding strategy~$\sigma$ and BPIC with
active block producers, then 
the payment rule $\prices$ is identically
zero on the range of~$\sigma$.
\end{theorem}
The proof of Theorem~\ref{t:main} will show that the result holds even
if 
BPs are restricted to have
nonnegative additive valuations and all known transactions can be included
simultaneously into a single feasible block.

\vspace{-.1in}

\paragraph{Discussion.}
The role of an impossibility result like Theorem~\ref{t:main} is to
illuminate the most promising paths forward.  From it, we learn that
our options are (i) constrained; and (ii) already being actively explored
by the blockchain research community.  Specifically, with active BPs,
to design a non-trivial TFM, we must choose from among three options:
\begin{enumerate}

\item Give up on ``good UX,'' at least as it is expressed by the DSIC
  property.  Arguably, this is the status quo, at least for blockchain
  protocols in which BPs are sufficiently motivated to be active.

\item Give up on the BPIC property, presumably compensating with
  restrictions on block producer behavior (perhaps enforced
  using, e.g., trusted hardware~\cite{flashbots2} or cryptographic
  techniques~\cite{encrypted_mempools}).

\item Expand the TFM design space, for example by incorporating
  order flow auctions (e.g., \cite{flashbots1}) or block producer
  competition (e.g., \cite{dom_mev_burn}) to expose information about
  a BP's private valuation to a TFM. (See Section~\ref{s:conc} for
  further discussion.)

\end{enumerate}

\vspace{.1in}

\begin{prevproof}{Theorem}{t:main}
Let $\tfm$ be a TFM that is BPIC, and DSIC with the bidding
strategy~$\sigma$.  By the Revelation Principle (see
footnote~\ref{foot:rp}), we can assume that~$\sigma$ is the truthful
bidding strategy (i.e., the identity function).
Toward a contradiction, suppose there is a nonnegative additive BP
valuation $\vbp$, a BP 
blockset $\blocks$, a set of transactions with bids $\bids$, and a
transaction~$t^*$ such that $\price_{t^*}(B,\bids) > 0$, where~$B =
\allocs(\bids,\vbp,\blocks)$.  Because the pricing rule~$\prices$
is individually rational (see Section~\ref{ss:tfm}), we must have $t^* \in
B$.  Because the TFM $\tfm$ is BPIC, the block~$B$ must maximize the BP's BPS
over all blocks in its blockset $\blocks$. 

We next define a modified BP valuation and a modified bid vector.
First, let $\bids'=(0,\bids_{-t^*})$ denote the bid vector in which 
the original bid $\bid_{t^*}$ for transaction~$t^*$
is dropped to~0 and all other bids are held fixed.
Second, let~$P$ denote the sum of the bids of all known transactions
(i.e., $P = \sum_t \bidt$) and $Q$ the burn that the TFM
would require on the modified bid vector for the block~$B$ (i.e., 
$Q = \burn(B,\bids')$), and define a modified
(but still additive) valuation $\vbphat$ so that
$\vbphat(\{t\}) > \vbp(\{t\}) + P + Q$ for all $t \in B$ and
$\vbphat(\{t\}) = 0$ for all $t \notin B$.  

The key claim is that the BPS-maximizing block
$\allocs(\bids',\vbphat,\blockset)$ for the modified valuation with
the modified bid vector contains every transaction of~$B$, and in
particular~$t^*$.  Under this modified valuation and bid vector, the
BPS of a block~$B' \in \blocks$ can be written as
\begin{equation}\label{eq:main0}
\vbphat(B') + \sum_{t \in B'} \pricet(B',\bids')
 - \burn(B',\bids').
\end{equation}
By the definition of~$\vbphat$, any transaction in~$B$ omitted
from~$B'$ results in a loss of more than $P+Q$ in the private valuation of the
BP:
\begin{equation}\label{eq:main1}
\vbphat(B) > \vbphat(B') + P + Q
\end{equation}
for every feasible block $B' \not\supseteq B$.
Next, individual rationality of the payment rule~$\prices$ implies
that the 
maximum revenue a BP can receive from including a transaction $t$ is
the attached bid $\bidt$, and thus the maximum revenue it receives
from any block in $\blocks$ is at most $P$.  Because the
payment rule~$\prices$ is nonnegative, we have
\begin{equation}\label{eq:main2}
\sum_{t \in B'} \pricet(B',\bids') \le
\sum_{t \in B} \pricet(B,\bids') + P
\end{equation}
for every~$B' \in \blocks$.
Finally, because the burning rule~$\burn$ is nonnegative,
\begin{equation}\label{eq:main3}
\burn(B,\bids') \le
\burn(B',\bids') + Q
\end{equation}
for every~$B' \in \blocks$.  Combining the
inequalities~\eqref{eq:main1}--\eqref{eq:main3} with~\eqref{eq:main0}
then implies that, with the modified valuation and bid vector, the BPS
of the block~$B$ is strictly higher than that of every block that
omits at least one of $B$'s transactions:
\[
\underbrace{\overbrace{\vbphat(B)}^{> \vbphat(B')+P+Q} + \overbrace{\sum_{t \in B}
    \pricet(B,\bids')}^{\ge 0}
 - \overbrace{\burn(B,\bids')}^{=Q}}_{\text{BPS of~$B$}}
>
\underbrace{\vbphat(B') + \overbrace{\sum_{t \in B'}
    \pricet(B',\bids')}^{\le P}
 - \overbrace{\burn(B',\bids')}^{\ge 0}}_{\text{BPS of~$B'$}}
\]
for every $B' \not\supseteq B$.  This completes the proof of the key
claim.

The point of this claim is that, when the BP has valuation $\vbphat$
and blockset $\blockset$ and the other transactions' bids are
$\bids_{-t^*}$, the transaction~$t^*$ will be included in the BP's
chosen block $B'=\allocs(\bids',\vbphat,\blocks)$ even when its
creator sets $\bid_{t^*}=0$.  Because the payment rule~$\prices$ is
individually rational, $\price_{t^*}(\bids',B')=0$.  Because the user
that created transaction~$t^*$ can guarantee inclusion at price~0 with
a bid of~0, any bid that leads to a positive price is automatically
suboptimal for it. Because the TFM $\tfm$ is DSIC with the truthful
bidding strategy, $t^*$ must be included at a price of~0 also when its
creator submits the original bid $\bid_{t^*}$; that is, if~$\hat{B}$
denotes $\allocs(\bids,\vbphat,\blocks)$, then $t^* \in \hat{B}$ and
$\price_{t^*}(\hat{B},\bids) = 0$.

We can complete the proof by arguing that $\hat{B} = B$.  (This would
imply that $\price_{t^*}(B,\bids) = 0$, in direct contradiction of our
initial assumption.)  By definition, the block~$B$ is a BPS-maximizing
block for a BP with valuation $\vbp$ and blockset $\blocks$ with
transaction bids~$\bids$, and it is the first such block with respect
to some fixed ordering over $\blocks$ (recall
Definition~\ref{d:consonant}(b)).  By construction of the modified
valuation $\vbphat$, passing from $\vbp$ to $\vbphat$ increases the
private value of the block~$B$ 
at least as much as any other block of $\blocks$.
Because the payment and burning rules of a TFM are independent of the
BP valuation, holding the bids~$\bids$ fixed, the block~$B$ also
enjoys at least as large a BPS increase as any other block of
$\blocks$.  Thus, the BPS-maximizing blocks with respect to the
modified valuation $\vbphat$ are a subset of those with respect to the
original valuation~$\vbp$, and this subset includes the block~$B$.
Because the allocation rule breaks ties consistently,
$\hat{B} = \allocs(\bids,\vbphat,\blocks)$ must be the original
block~$B$.
\end{prevproof}

\begin{remark}[Variations of Theorem~\ref{t:main}]\label{rem:main}
Variations on the proof of Theorem~\ref{t:main} above show that the
same conclusion holds for:
\begin{itemize}

\item [(a)] BPs that have a non-zero private value for only one block
  (a very  special case of single-minded valuations).  This
version of the argument does not require the consistent tie-breaking
assumption in Definition~\ref{d:consonant}(b).

\item [(b)] Burning rules that need not be nonnegative (i.e.,
  rules that can print money), provided
  that, for every bid vector $\bids$, there is a finite lower bound on
  the minimum-possible burn $\min_{B \in \blocks} \burn(B,\bids)$.  (This
  would be the case if, for
  example, the blockset~$\blocks$ is finite.)

\item [(c)] Bid spaces and payment rules that need not be
  nonnegative (i.e., with negative bids and user rebates allowed,
  subject to individual rationality),
  provided there is a finite minimum bid~$\bid_{min} \in (-\infty,0]$
  and that $\pricet(B,\bids) = \bid_{min}$ whenever $t \in B$ with
  $\bidt = \bid_{min}$.  In this case, the argument shows that the
  payment rule $\prices$ must be
  identically equal to $\bid_{min}$ on the range of~$\sigma$.

\end{itemize}
\end{remark}

\subsection{Marginal Values and Approximate DSIC Guarantees}

Non-trivial TFMs satisfying the DSIC and BPIC conditions exist with
passive BPs (Example~\ref{ex:tipless}) but not with active BPs
(Theorem~\ref{t:main}). This section aims for a more quantitative
understanding of this difference, and its main result shows a sense in
which the ``amount by which DSIC fails'' for a user scales with the
amount of value that a BP can extract from that user's transaction.
For this result, we assume that a BP's blockset is {\em
  downward-closed}, meaning that removing a transaction from a
feasible block yields another feasible block (e.g., if feasibility is
determined by an upper bound on the total size of the included
transactions).  

In general, a BP's value for a transaction depends on
the other transactions in the block. The next definition quantifies
the best-case marginal value of a transaction to a BP.
\begin{definition}[Maximum Marginal Value of a Transaction]\label{d:marginal}
With respect to a BP valuation~$\vbp$ and a downward-closed
blockset~$\blocks$, the {\em maximum marginal value}~$\nu_t$ of a transaction~$t$ is 
\[
\nu_t := \max_{B \in \blocks \,:\, t \in B} \left[ \vbp(B) - \vbp(B \sm
\{t\}) \right],
\] 
where ``$B \sm \{t\}$'' denotes the block~$B$ with the transaction~$t$
removed.
\end{definition}
For a passive BP, the maximum marginal value of a transaction is always~0.  If
a BP has an additive valuation (defined by transaction-specific
constants~$\mu_t = \vbp(\{t\})$),
the maximum marginal value~$\nu_t$ of a transaction is the same as the
BP's stand-alone value $\mu_t$ for it.

We posit that, at least in some scenarios, even though a BP's
valuation is private, a user may know a reasonable upper bound on the
maximum marginal value of its transaction.  For example, a direct
payment from one user to another might reasonably be assumed to always
have~0 marginal value to the BP.  For an AMM trade, it may be possible
to estimate the value that could be extracted by the BP from front-
and/or backrunning the transaction.

We next show that, in a BPIC version of the tipless mechanism
(Example~\ref{ex:tipless}), an upper bound on the maximum marginal
value of a user's transaction~$t$ is also an upper bound on both the
amount by which the user might want to deviate from its recommended
bid and the additional utility the user stands to gain from such a
deviation.\footnote{The same result holds for the EIP-1559 mechanism
  (see Example~\ref{ex:1559}) when the base fee is not excessively low
  (see Example~\ref{ex:15592}).}  Precisely, the {\em BPIC tipless
  mechanism} uses the same payment and burning rules as in
Example~\ref{ex:tipless} and redefines the allocation rule to restore
consonance (instructing the BP to select a BPS-maximizing block,
breaking ties according to some fixed ordering over blocks).  
\begin{theorem}[The BPIC Tipless Mechanism Is Approximately DSIC]\label{t:dsic_approx}
In the BPIC tipless mechanism with base fee $\basefee$, for the
bidding strategy~$\sigma(\valt) = \min \{ \valt, \basefee
\cdot s_t \}$, and assuming that the BP has a downward-closed blockset and
carries out the suggested allocation rule:
\begin{itemize}

\item [(a)] for a user with transaction~$t$ and valuation~$\valt$, 
bidding more than $\sigma(\valt)$ is dominated by bidding
$\sigma(\valt)$; 

\item [(b)] for a user with transaction~$t$, valuation~$\valt$, and
  maximum marginal value~$\nu_t$, bidding less than $\sigma(\valt) - \nu_t$ is
  dominated by bidding $\sigma(\valt)$.

\end{itemize}
\end{theorem}

\paragraph{Discussion.}
Theorem~\ref{t:dsic_approx} shows that the range of undominated
strategies for a user, and therefore also the maximum additional
utility it could earn from deviating from its recommended bid, is no
larger than the maximum marginal value~$\nu_t$ of its
transaction~$t$. (At best, bidding~$y$ less than the recommended
bid~$\sigma(\valt)$ in the tipless mechanism decreases a user's
payment by~$y$ without affecting its transaction's inclusion.)  This
is good news on two counts.  First, if marginal values are always
small, the ``obvious'' bidding strategy is near-optimal.  Second, this
guarantee holds user-by-user; for example, a user that knows that its
own transaction has zero marginal value (e.g., a direct payment) can
follow the recommended bidding strategy without any second-guessing,
even without knowing anything about the other transactions and their
maximum marginal values.

\vspace{.1in}

\begin{prevproof}{Theorem}{t:dsic_approx}
Fix a base fee~$r$, a BP valuation and blockset, a transaction~$t$, a
valuation~$\valt$, and bids $\bidsmt$ for all transactions other
than~$t$. Let~$\nu_t$ denote the maximum marginal value of~$t$.

To prove part~(a), we consider two cases.  First, suppose that
$\valt \ge r \cdot s_t$, in which case the recommended bid is
$\sigma(\valt) = \basefee \cdot s_t$.  
Every overbid $\bidt > \basefee \cdot s_t$
leads to the same outcome as the bid~$\basefee \cdot s_t$. For
example, the fee paid and the required burn for the transaction (if
included) both remain equal to $r \cdot s_t$. Further, the set of
BPS-maximizing blocks is identical with the bids $\bidt$ and~$r \cdot
s_t$, so the transaction~$t$ is included in the BP's chosen
block either in both cases or in neither case. (Recall from
Example~\ref{ex:tipless} that the
tipless mechanism breaks ties according to
some fixed ordering over feasible blocks.)  In the
second case, with $\valt < \basefee \cdot s_t$, the recommended bid
is~$\sigma(\valt) = \valt$.  If the creator of~$t$ bids
$\bidt > \valt$ and~$t$ is excluded from the BP's block, then~$t$ would
also be excluded from the BP's block with a bid of $\valt$.  If the
BP does include~$t$ with a bid of $\bidt$, then its creator must pay
$\min\{ \bidt,r\cdot s_t\} > \valt$, resulting in negative
utility. (Whereas the recommend bid~$\sigma(\valt)=\valt$ would 
guarantee nonnegative utility.)

For part~(b), we claim that a bid lower than $\sigma(\valt)-\nu_t$
guarantees that the BP will exclude the transaction~$t$ from its
chosen block, resulting in zero user utility.
Because~$\sigma(\valt) \le \basefee \cdot s_t$ and the burn required
by the BP for including the transaction is $\basefee \cdot s_t$, a bid
lower than $\sigma(\valt)-\nu_t$ forces the BP to pay a subsidy larger
than $\nu_t$ to include~$t$ in its block.  Because the BP enjoys at
most~$\nu_t$ additional private value from~$t$'s inclusion (by the
definition of the maximum marginal value~$\nu_t$) and its blockset is
downward-closed (by assumption), the BP is always better off
excluding~$t$ than including it.  Meanwhile, because
$\sigma(\valt) \le \valt$ and the payment rule is individually
rational, following the recommended bidding strategy always guarantees
nonnegative user utility.
\end{prevproof}

\section{The Welfare Achieved by DSIC and BPIC Mechanisms}

Theorem~\ref{t:main} shows that TFMs that are DSIC and BPIC must be
``trivial,'' in the sense that users are never charged for the
privilege of transaction inclusion.  The next result formalizes the
intuitive consequence that such TFMs may, if both users and the BP
follow their recommended actions, produce blocks with welfare
arbitrarily worse than the maximum possible.  (Recall that the
welfare~$W(B)$ of a block~$B$ is defined in
expression~\eqref{eq:welfare} in Section~\ref{ss:welfare}.)  That is,
no approximately welfare-maximizing TFM can be both DSIC and BPIC with
active BPs.  This result is not entirely trivial because the
conclusion of Theorem~\ref{t:main} imposes no restrictions on the
burning rule of a TFM.

\begin{theorem}[Impossibility of Non-Trivial Welfare Guarantees]\label{t:welfare}
Let $\tfm$ denote a TFM that is BPIC and DSIC with bidding
strategy~$\sigma$.
For every approximation factor~$\rho >
0$, there exists a BP valuation~$\vbp$, BP blockset~$\blocks$, 
block~$B^* \in \blocks$, and transactions with corresponding user
valuations $\vals$ such that
\[
W(B) \le \rho \cdot W(B^*),
\]
where~$B = \allocs(\sigma(\vals),\vbp,\blocks)$.
\end{theorem}

\begin{proof}
Let~$\tfm$ denote a TFM that is DSIC and BPIC.
By Theorem~\ref{t:main}, the
payment rule $\prices$ is identically zero on the range of its
recommended bidding strategy~$\sigma$.  We assume that (appealing to DSIC)
users always follow this bidding strategy~$\sigma$
and that (appealing to BPIC) the BP always chooses the block
recommended by the allocation rule~$\allocs$.
By the Revelation Principle (see footnote~\ref{foot:rp}), we can
assume that~$\sigma$ is the identity function.

Suppose there are two known transactions, $y$ and $z$, with arbitrary
positive user valuations~$\val_y$ and~$\val_z$.  Suppose the BP
blockset $\blocks$ comprises three
feasible blocks, $B_0 = \{\}$, $B_y = \{y\}$, and $B_z = \{z\}$.  Set
$\vbp(B_0) = \vbp(B_y) = 0$ and
\[
\vbp(B_z) = \burn(B_z,\vals)+\eps
\]
for some small $\eps > 0$.  Then, because the burning
rule~$\burn$ is nonnegative and the payment rule $\prices$ is
identically zero, the BP will choose the block
$B_z$ (i.e., $\allocs(\vals,\vbp,\blocks)=B_z$).

To complete the proof, we range over all valuation vectors of the form
$\vals'=(\val'_y,\val_z)$ 
and treat separately three
(non-exclusive but exhaustive) cases:
\begin{itemize}

\item [(C1)] Every choice of~$\val'_y$ leads the BP to choose~$B_z$
 (i.e., $\allocs(\vals',\vbp,\blocks)=B_z$ for all $\val'_y$).

\item [(C2)] Some choice of~$\val'_y$ leads the BP to choose~$B_y$
 (i.e., $\allocs(\vals',\vbp,\blocks)=B_y$).

\item [(C3)] Some choice of~$\val'_y$ leads the BP to choose the empty
block (i.e., $\allocs(\vals',\vbp,\blocks)=B_0$).

\end{itemize}

In case~(C1), because the BP always, no matter the value of $\val'_y$,
chooses the block~$B_z$ (with welfare $\val_z + \burn(B_z,\vals)+\eps$)
over the block~$B_y$ (with welfare $\val_y'$), letting $\val_y'$ tend
to infinity 
proves the desired result (with~$B=B_z$ and~$B^*=B_y$).  

Case~(C2) cannot occur, for if it did, the creator of transaction~$y$
would prefer to misreport its valuation (as $\val_y'$) when its true
valuation is $\val_y$, contradicting the assumption that the TFM
$\tfm$ is DSIC with the truthful bidding strategy. (Because~$\prices$
is identically~0 and~$\val_y > 0$, the creator of~$y$ always strictly
prefers inclusion to exclusion.)

Finally, in case~(C3), the result follows immediately from the facts
that~$W(B_0)$ is zero while~$W(B_y)$ and $W(B_z)$ are positive.
\end{proof}

\begin{remark}[Generalizations of Theorem~\ref{t:welfare}]
The proof of Theorem~\ref{t:welfare} shows that the result holds
already with
BPs that have additive or single-minded valuations. (As discussed in
Remark~\ref{rem:main}, Theorem~\ref{t:main} holds in both these cases,
and the BP valuation~$\vbp$ used in the proof of
Theorem~\ref{t:welfare} is both additive and single-minded).
A slight variation of the proof shows that the result holds more
generally for DSIC and BPIC TFMs that use a not-always-nonnegative
burning rule, under the same condition as in
Remark~\ref{rem:main}(b).
\end{remark}

\begin{remark}[Welfare Guarantees Under Stronger Assumptions]\label{rem:welfare_approx}
The root issue driving the proof of Theorem~\ref{t:welfare} is
that, in the absence of user payments, a BP may prefer a block of
low-value transactions from which it can extract a small amount of
value over a block of high-value transactions from which it cannot
extract any value.  Said differently, the preferences of users and the
BP over blocks may be very different, and without user payments
there's no tool available to align their interests.

Severely misaligned preferences, with user valuations dwarfing the
BP's valuation, are necessary for the impossibility result in
Theorem~\ref{t:welfare} to hold.  For example, suppose that the
intensity of the BP's preferences is at least somewhat commensurate
with that of the users, in the sense that, for some
parameter~$\beta > 0$,
\begin{equation}\label{eq:welfare_approx}
\vbp\left(B^{BP}\right) \ge \beta \cdot \sum_{t \in B^u} \valt,
\end{equation}
where~$B^{BP}$ is the block of~$\blocks$ that maximizes the BP's
private value and~$B^u$ is the block of~$\blocks$ that maximizes the
total user value of the included transactions.
In this case, the trivial (DSIC and BPIC) TFM $\tfm$ in which $\prices$
and~$\burn$ are identically zero, and with~$\allocs$ instructing the
BP to choose a BPS-maximizing block, will in fact produce a block
with welfare at
least $\beta/(\beta+1)$ times that of a welfare-maximizing block~$B^*$ 
of~$\blocks$.  (Proof sketch: combine~\eqref{eq:welfare_approx}
with the inequality $W(B^*) \le \vbp(B^{BP}) + \sum_{t \in B^u} \valt$
and rearrange.)
\end{remark}

\section{Future Directions}\label{s:conc}

What are the minimal additions to the model of TFMs studied in this
paper that allow blockchain transactions to be priced and processed in
a way that is incentive-compatible for all participating parties?  To
fire up the reader's imagination, we next summarize two of the many
ideas that are currently being actively investigated by the blockchain
research community (one cryptographic and the other economic).

\vspace{-.1in}

\paragraph{Encrypted mempools.} 
The rough idea of an encrypted mempool is to require users to 
submit transactions to BPs in an encrypted form. Transactions are
decrypted only after publication, so that a BP must commit to its
block without knowing the actual content of the included transactions.
In an ideal implementation, every transaction 
would ``look identical'' to the BP, effectively transforming it from
an active BP 
into a
passive one,
which in turn would
circumvent the impossibility result in Theorem~\ref{t:main}.  There
are a number of obstacles to achieving such an ideal implementation,
of course, such as the possibility of side information about
transactions gleaned by a BP from other sources.

\vspace{-.1in}

\paragraph{Competition among block producers.}
A different approach to restoring incentive-compatibility to TFMs with
active BPs is to force multiple BPs to compete with each other, either
for individual transactions or for entire blocks.  For example,
instead of a single monopolist BP, a blockchain protocol could
designate two or more candidate BPs, who would then bid for the right
to publish the next block.  The hope is that such competition would
lead to price discovery, with the winning bid revealing a good
approximation of the winning BP's private value for the block that it
published.  Price discovery would effectively expose the winning BP's
private valuation to the TFM, which could then be used as another
input to the TFM's payment and burning rules. This idea may suffice
to evade the impossibility result in Theorem~\ref{t:main}, the proof
of which relies on the assumption that a TFM has no knowledge about a
BP's private valuation.

\vspace{.1in}

We leave it as future work to flesh out the details of these and other
approaches to expanding the TFM design space and determine whether
they are indeed sufficient for the existence of non-trivial DSIC and
BPIC mechanisms.

\end{document}